\documentclass[12pt]{article}
\usepackage{amsgen,amsmath,amstext,amsbsy,amsopn,amssymb}
\usepackage[dvips]{graphicx}
\usepackage{float}
\usepackage[numbers]{natbib}
\usepackage{hyperref}  
\setcitestyle{authoryear,open={(},close={)}}
\usepackage[utf8]{inputenc}
\usepackage[T1]{fontenc}

\usepackage{subcaption}
\RequirePackage{bm}
\usepackage{epsfig}
\usepackage{fancybox}
\usepackage{tikz}
\usepackage{bbold}
\usepackage{indentfirst}
\usepackage{amsthm}
\usepackage{booktabs}

\usepackage{algorithm}  
\usepackage{algpseudocode}  
\usepackage{amsmath}  
\usepackage{multirow}
\usepackage{makecell} 

\renewcommand{\baselinestretch}{2}  

\newtheorem{theorem}{\bf Theorem}

\newtheorem{innercustomgeneric}{\customgenericname}
\providecommand{\customgenericname}{}
\newcommand{\newcustomtheorem}[2]{%
  \newenvironment{#1}[1]
  {%
   \renewcommand\customgenericname{#2}%
   \renewcommand\theinnercustomgeneric{##1}%
   \innercustomgeneric
  }
  {\endinnercustomgeneric}
}

\newcustomtheorem{customthm}{Theorem}
\newcustomtheorem{customlemma}{Lemma}

\newcommand{\xb}{\mathbf{x}}

\newcommand{\bx}{\bm{x}}

\newcommand{\Kb}{\mathbf{K}}

\newcommand{\Xb}{\mathbf{X}}
\newcommand{\Yb}{\mathbf{Y}}

\newcommand{\bX}{\bm{X}}

\newcommand{\0}{{\mathbf{0}}}

\def\bX{{\bf X}}

\def\bx{{\bf x}}

\def\b0{{\bf 0}}

\usepackage{graphicx} 
\usepackage{booktabs} 
\usepackage{array}
\usepackage{rotating}
\usepackage{lscape}
\usepackage{xcolor}
\begin{document}

\def\spacingset#1{\renewcommand{\baselinestretch}%
{#1}\small\normalsize} \spacingset{1}

\title{\bf{\large A General Framework of Nonparametric Feature Selection in High-Dimensional Data }}
\author{\normalsize Hang Yu, Yuanjia Wang, and Donglin Zeng}
\date{}
\bigskip

\begin{titlepage}

\maketitle

\begin{abstract}
Nonparametric feature selection in high-dimensional data is an important and challenging problem in statistics and machine learning fields. Most of the existing methods for feature selection focus on parametric or additive models which may suffer from model misspecification. In this paper, we propose a new framework to perform nonparametric feature selection for both regression and classification problems. In this framework, we learn prediction functions through empirical risk minimization over a reproducing kernel Hilbert space. The space is generated by a novel tensor product kernel which depends on a set of parameters that determine the importance of the features. Computationally, we minimize the empirical risk with a penalty to estimate the prediction and kernel parameters at the same time. The solution can be obtained by iteratively solving convex optimization problems. We study the theoretical property of the kernel feature space and prove both the oracle selection property and the Fisher consistency of our proposed method. Finally, we demonstrate the superior performance of our approach compared to existing methods via extensive simulation studies and application to a microarray study of eye disease in animals.
\end{abstract}

\noindent%
\vfill

\let\thefootnote\relax\footnotetext{Hang Yu is PhD candidate at Department of Statistics and Operation Research, University of North Carolina, Chapel Hill, NC 27599 (hangyu@live.unc.edu),
Yuanjia Wang is Professor at Department of Biostatistics, Columbia University, New York, NY 10032 (yw2016@cumc.columbia.edu), and Donglin Zeng is Professor at Department of Biostatistics, University of North Carolina, Chapel Hill, NC 27599 (dzeng@email.unc.edu). 
This research is supported by U.S. NIH grants NS073671, GM124104, and MH117458.}

\end{titlepage}

\newpage
\clearpage
\setcounter{page}{1}
\spacingset{1.55} 

\section{Introduction}
With  biotechnology advances in modern medicine, biomedical studies collecting complex data with a large number of features are becoming the norm. High-dimensional feature selection is an essential tool to allow using such data for disease prediction or precision medicine,  for instance, to discover a set of diagnostic biomarkers from neuroimaging measures for early prediction of neurodegenerative diseases, or to determine predictive  biomarkers for  effective management of type 2 diabetic patients' healthcare. Accurately identifying the subset of true important features is even more crucial and challenging than before in the fields of statistics and machine learning.

High-dimensional feature selection has been extensively studied for linear or generalized linear models in the past decades, and many methods have been developed including Lasso \citep{tibshirani1996lasso}, SCAD \citep{fan2001scad}, {
MCP \citep{zhang2010MCP} and \citep{wang2012parametric}. In these parametric models, the importance of individual features is characterized by non-null coefficients associated with them, so proper penalization can identify those non-null  coefficients with probability tending to one when the sample size increases. However, parametric model assumptions are likely to be incorrect for many biomedical data due to potential correlations and higher-order interactions among feature variables. In fact, applying these approaches to any simple transformation of feature variables may lead to very different feature selection results.

More recently,  increasing efforts have been devoted to  high-dimensional feature selection when parametric assumptions, especially linearity assumption, do not hold. Various approaches were proposed to select features based on measuring certain marginal dependency
(\citet{guyon2003selection}, \citet{fan2008sure}, \citet{fan2011nonparametric}, \citet{song2012dependence}, \citet{yamada2014HSIC},  \citet{urbanowicz2018rbf}).
For example,  nonparametric association between each feature and outcome was used for screening (\citet{fan2008sure}, \citet{fan2011nonparametric}, \citet{song2012dependence}).
\citet{RRCS} adopted a a robust rank correlation screening  method based on  marginal Kendall correlation coefficient. \citet{yamada2014HSIC} considered a feature-wise kernelized Lasso, namely HSICLasso, for capturing nonlinear dependency between features and outcomes. In this approach, after a Lasso-type regression of an output kernel matrix on each feature-wise kernel matrix, unimportant features with small marginal dependence in terms of a Hilbert-Schmidt independence criterion (HSIC) would be removed.
 However, all  methods based on marginal dependence may fail to select truly important variables since marginal dependency does not necessarily imply the significance of a feature when other features are also included for prediction, which is the case even for a simple linear model.
 
Alternatively, other approaches were proposed to relax parametric model assumptions and perform feature selection and prediction simultaneously. 
 \citet{Lin2006cosso} proposed COmponent Selection and Smoothing Operator (COSSO) to perform penalized variable selection based on smoothing spline ANOVA.  \citet{SPAM} studied feature selection in a sparse additive model (SpAM),  which assumed an additive model but allowed arbitrary nonparametric smoothers such as approximation in a reproducing kernel Hilbert space (RKHS) for each individual component function. \citet{Huang2010} considered spline approximation in the same model and adopted an adaptive group Lasso method to perform feature selection. Although both COSSO and SpAM allowed nonlinear prediction from each feature, they still imposed restrictive additive model structures, possible with some higher-order interactions. To allow arbitrary interactions among the features and perform a fully nonparametric prediction, \citet{allen2013kernel} propsed a procedure named KerNel Iterative Feature Extraction (KNIFE), in which the feature input was constructed in a Gaussian RKHS in order to perform nonparametric prediction. Different weights were used for different features in the constructed Gaussian kernel function so that a larger weight implied a higher importance of  the corresponding feature variable. However, due to high nonlinearity in the kernel function, estimating the weights was numerically unstable even when the dimension of the features was moderate.

In this paper, we propose a general framework to perform nonparametric high-dimensional feature selection. We consider a general loss function which includes both regression models and classification as special cases. To perform nonparametric prediction, we construct a novel RKHS based on a tensor product of kernels for individual features. The constructed tensor product kernel, as discussed in \citet{gao2012iwiee}, can handle any high-order nonlinear relationship between the features and outcome and any high-order interactions among the features. More importantly, each feature kernel depends on a non-negative parameter which determines the feature importance, so for feature selection, we further introduce a $l_1$-penalty of these parameters in the estimation. 
Computationally, coordinate descent algorithms are used for updating parameters and each step involves simple convex optimization problems.
Thus, our algorithm is numerically stable and can handle high-dimensional features easily.
Theoretically, we first derive the approximation property of the proposed RKHS and characterize the complexity of the unit ball in this space in terms of bracket covering numbers. We then show that the estimated prediction function from our approach is consistent and moreover, we show that under some regularity conditions, the important features can be selected with probability tending to one.

The rest of the paper is organized as follows. In Section 2, we introduce our proposed regularized tensor product kernel and lay out a penalized framework for both estimation and feature selection. We then provide detailed  computational algorithms to solve the optimization problem. In Section 3, we provide two theorems studying the property of the proposed RKHS. We then give the main result of this paper including the consistency of the estimated prediction function and the oracle property of the feature selection. In Section 4, two simulation studies for regression and classification problems are conducted and we compare our method to existing methods.  Application to a microarray study is given in Section 5. We conclude the paper with some discussion in Section 6.

\section{Method}

Suppose data are obtained from $n$ independent subjects  and consist of  $(\Xb_i, Y_i), i=1,...,n$, where we let $\Xb$ denote $p_n$-dimensional feature variables and $Y$ be the outcome which can be continuous, binary or ordinal. Our goal is to use the data to learn a nonparametric prediction function, $f(\Xb)$, for the outcome $Y$.

We  learn $f(\Xb)$  through a regularized empirical risk minimization by assuming $f(\cdot)$ belongs to a RKHS associated with a kernel function, $\kappa(\Xb, \widetilde{\Xb})$, which will be described later. Specifically, if we denote
the RKHS generated by  $\kappa(\Xb, \widetilde {\Xb})$  by $\mathcal{H}_{\kappa}$, equipped with norm  $\Vert \cdot\Vert_{\mathcal{H}_{\kappa}}$, then 
the empirical regularized risk minimization on RKHS for estimating $f(\Xb)$ solves the following optimization problem:
$$
\min_{f} {\bf P}_n l(Y,f(\Xb)) +\gamma_n\Vert f\Vert_{\mathcal{H}_{\kappa}}^2,
$$
where $l(y,f)$ a pre-specified non-negative and convex loss function to quantify the prediction performance, 
${\bf P}_n$ denotes the empirical measure from $n$ observations, i.e., ${\bf P}_ng(Y, \Xb)=n^{-1}\sum_{i=1}^n 
g(Y_i, \Xb_i)$, and
$\gamma_n$ is a tuning parameter to control the complexity of $f$.  For a continuous outcome, $l(y, f)$ is often chosen to be a $L_2$-loss given as $(y-f)^2$, while for a binary outcome, it can be one of the large-margin losses such as $\exp\{-yf\}$ in Adaboost.
There are many choices of kernel functions for $\kappa(\cdot, \cdot)$ so that the estimated $f(\Xb)$ is nonlinear.
One of the most commonly used kernel functions in machine learning is the Gaussian kernel function given by 
$\kappa(\Xb, \widetilde{\Xb})=\exp\left\{-\Vert \Xb-\widetilde{\Xb}\Vert^2/\sigma^2\right\}$ for some bandwidth $\sigma$, where $\Vert\cdot\Vert$ is the Euclidean norm.
To handle high-dimensional features, SpAM considered an additive kernel function by 
assuming $\kappa(\Xb, \widetilde{\Xb})=\sum_{j=1}^{p_n}\exp\left\{-\vert X_j-\widetilde X_j\vert^2/\sigma^2\right\}$. In the KNIFE procedure, the kernel function  is defined as $\kappa_{\bm{\omega}}(\Xb, \widetilde {\Xb}) = \exp\left\{-{\sum_{j=1}^{p_n} \omega_j(X_j -\widetilde{X}_j)^2}/{\sigma^2}\right\} $, where $\omega_j, j=1,.., p_n$ are the additional weights to determine the feature importance.

To achieve the goal of both nonparametric prediction and feature selection, we propose a tensor product kernel as follows. For any given nonnegative vector $\bm{\lambda} = (\lambda_1, \lambda_2, \cdots \lambda_{p_n})^\intercal$, we define a $\bm{\lambda}$-regularized kernel function  as
\begin{equation}
 \kappa_{\bm{\lambda}, \sigma_n}(\Xb,\widetilde\Xb) = \prod_{m=1}^{p_n}\left\{1+\lambda_m \kappa_n(X_{m}, \widetilde X_{m})\right\},
 \end{equation}
 where $\kappa_n(x,y)=\exp\left\{-(x-y)^2/2\sigma_n^2\right\}$ with a pre-defined bandwidth $\sigma_n$ in $\mathcal{R}$. 
 There are two important observations for this new kernel function. First,
it is a product of a univariate kernel function for each feature variable, which is given by $1+\lambda_m \kappa_n(X_m, \widetilde X_m)$. Thus, the RKHS generated by $\kappa_{\bm{\lambda}, \sigma_n}$ is equivalent to the tensor product of the RKHS generated by each feature-specific space.  Second, each univariate kernel function is essentially the same as the Gaussian kernel function when $\lambda_m\neq 0$. Consequently, the resulting tensor product space is the same as the RKHS generated by the multivariate Gaussian kernel function from all features whose $\lambda_m$'s are non-zero. Therefore, the closure for the RKHS generated by $\kappa_{\bm{\lambda}, \sigma_n}$ consists of all functions that only depend on feature variables for which $\lambda_m\neq 0$. In other words, 
non-negative parameters, $\lambda_m$, completely capture and regularize the contribution of  each feature $X_m$. 
In this way, the feature selection can be achieved by estimating the regularization parameters, $\lambda_m$'s, in the kernel function.

More specifically, using the proposed kernel function, we let ${\mathcal{H}_{\bm{\lambda}, \sigma_n}}$ denote the RKHS corresponding to  $\kappa_{\bm{\lambda}, \sigma_n}$ so we aim to minimize
\begin{equation}
\begin{aligned}
&L_n(\bm{\lambda}, f) \equiv && {\bf P}_n l(Y,f(\Xb)) +\gamma_{1n}||f||_{\mathcal{H}_{\bm{\lambda}, \sigma_n}}^2 +\gamma_{2n} P(\bm{\lambda})\\
&\text{subject to } && M\ge  \lambda_1, \lambda_2, \cdots, \lambda_{p_n} \geq 0,
\end{aligned}
\end{equation}
where $M$ is a pre-specified large constant.  $P(\bm{\lambda}) =\sum_{m=1}^{p_n} P(\lambda_m)=\sum_{m=1}^{p_n} \lambda_m I (\lambda_m < M/2),$ which is a truncated Lasso, and $\gamma_{1n}$, $\gamma_{2n} $ are tuning parameters. Here, we include an $l_1$ penalization term on the regularization vector to perform feature selection and restrict $\lambda_m$ to be bounded. The latter bound is useful for numerical convergence to avoid the situation that some $\lambda_m$ can diverge.
 Since our RKHS contains constant and based on the representation theory for RKHS, solution for (2) takes form  $$ f (\Xb)=\sum_{i=1}^n\alpha_i \kappa_{\bm{\lambda}, \sigma_n}(\Xb, \Xb_i) $$ and
$$\Vert f\Vert_{\mathcal{H}_{\bm{\lambda},\sigma_n}}^2=
\bm{\alpha}^T\Kb_{\bm{\lambda}, \sigma_n} \bm{\alpha},$$
where $\bm{\alpha}=(\alpha_1,...,\alpha_n)^T$ and $\Kb_{\bm{\lambda}, \sigma_n}$ is an $n\times n$ matrix with entry
$\kappa_{\bm{\lambda}, \sigma_n}(\Xb_i,\Xb_j).$ Then
the optimization becomes solving
\begin{equation*}
\begin{aligned}
&\min_{\alpha_1,...,\alpha_n, \bm{\lambda}}&
& {\bf P}_n l(Y,\sum_{i=1}^n\alpha_i \kappa_{\bm{\lambda}, \sigma_n}(\Xb, \Xb_i) ) +\gamma_{1n}\bm{\alpha}^T\Kb_{\bm{\lambda}, \sigma_n} \bm{\alpha}+ \gamma_{2n} \sum_{m=1}^{p_n}  \lambda_m I(\lambda_m < M/2)\\
&\text{subject to } && M\ge \lambda_1, \lambda_2, \cdots, \lambda_{p_n} \geq 0.
\end{aligned}
\end{equation*}

We iterate between $\bm{\alpha}$ and $\bm{\lambda}$ to solve the above optimization problem. At the $k$-th iteration,
\begin{align}
&\bm{\alpha}^{k+1}  =  \min\limits_{\bm{\alpha}} n^{-1}\sum_{j=1}^n l(Y_j,\sum_{i=1}^n\alpha_i \kappa_{\bm{\lambda}^{k} , \sigma_n}(\Xb_j, \Xb_i) ) + \gamma_{1n} \bm{\alpha}^\intercal \Kb_{\bm{\lambda}^k, \sigma_n} \bm{\alpha} \\
& \bm{\lambda}^{k+1} = \min\limits_{0\le \bm{\lambda}\le M} n^{-1}\sum_{j=1}^n l(Y_j,\sum_{i=1}^n\alpha_i^{k+1} \kappa_{\bm{\lambda}, \sigma_n}(\Xb_j, \Xb_i)  ) \notag \\
&\qquad \qquad\qquad + \gamma_{1n} ({\bm{\alpha}^{k+1}})^\intercal \Kb_{\bm{\lambda}, \sigma_n} \bm{\alpha}^{k+1}  + \gamma_{2n} \sum_{m=1}^{p_n}  \lambda_m I(\lambda_m < M/2).
\end{align}
Since the loss function is a convex loss, the optimization in (3) is a convex minimization problem, so many optimization algorithms can be applied. To solve (4) for $\bm{\lambda}$, we adopt a coordinate descent algorithm to update each $\lambda_q \ (q = 1, 2, \cdots, p_n)$ in turn. Specifically,  to obtain $\lambda_q^{k+1}$, we fix $\lambda_1^{k+1}, \lambda_2^{k+1}, \cdots, \lambda_{q+1}^k, \lambda_{q+2}^{k}, \cdots, \lambda_{p_n}^k$ and then after simple calculation, the objective function takes the following form,
\begin{equation}
 \min\limits_{\lambda_q\ge 0} \frac{1}{n} \sum_{i=1}^n g(a_{iq} +b_{iq}\lambda_q) + d_q \lambda_q,
\end{equation}
where 
$g(\lambda_q)$ is equal to $l(Y_j,\sum_{i=1}^n\alpha_i^{k+1} \kappa_{\bm{\lambda}, \sigma_n}(\Xb_j, \Xb_i)  )$ as  a function of $\lambda_q$, and 
$a_{iq}, b_{iq}, d_q$'s are constants. By the construction of $\kappa_{\bm{\lambda}, \sigma_n}$, $g(\lambda_q)$ is a convex function so each step in the coordinating descent algorithm is a constrained convex minimization problem in a bounded inteval, which is easy to solve.
 Thus, our algorithm guarantees that the objective function decreases over iterations and converges to a local minimum. We summarize the algorithm in the following table. At the convergence after $k$ iterations, the final prediction function is given as  
$$\widehat f_{\bm{\lambda}^{k+1}}(\Xb)=\sum_{i=1}^n {\alpha}_i^{k+1}\kappa_{\bm{\lambda}^{k+1}, \sigma_n}(\Xb, \Xb_i).$$
For classification problem,  the classification rule is $\textrm{sign}(\widehat f_{\bm{\lambda}^{k+1}}(\Xb))=\textrm{sign}(\sum_{i=1}^n {\alpha}_i^{k+1}\kappa_{\bm{\lambda}^{k+1}, \sigma_n}(\Xb, \Xb_i) ).$ We give details of our algorithm below (Algorithm 1).

\begin{algorithm}[H]  
  \caption{ Algorithm for learning $f(\Xb)$}  
  \begin{algorithmic}
    \Require  
      Data $(\Xb, \Yb)$; Regularization parameter $\gamma_{1n}$ and $\gamma_{2n}$;  Former updating results, $\widehat{\bm{\alpha}}^{k}, \widehat{\bm{\lambda}}^{k}, \widehat{f}_{\widehat{\bm{\lambda}}^{k}}$;  

      \hspace{-1.25cm}\textbf{Initialize} For regression, $ \widehat{\bm{\lambda}}_0 = \0$; For classification, $ \widehat{\bm{\lambda}}_0 = (0, \cdots, 1,\cdots, 0) $, where all elements equal to 0, expect the one having largest margin correlation with outcome.
      
      \hspace{-1.25cm}\textbf{Iterate} until convergence ($\delta = \vert L_n(\widehat{\bm{\lambda}}^{k+1}, \widehat{f}_{\widehat{\bm{\lambda}}^{k+1}}) - L_n(\widehat{\bm{\lambda}}^{k}, \widehat{f}_{\widehat{\bm{\lambda}}^{k}} ) \vert \leq c_1$, $e  ={\Vert  \widehat{\bm{\lambda}}^{k+1} - \widehat{\bm{\lambda}}^{k} \Vert}_1\leq  c_2 $, where $c_1$ and $c_2$ are given cut points):
      \begin{itemize}
      \item[(i)]  Update $\widehat{\bm{\alpha}}^{k+1}$ for fix $\widehat{\bm{\lambda}}^k$, which can be solved explicitly for regression and via fminsearch function for classification.
      \item[(ii)]  Update $\widehat{\bm{\lambda}}^{k+1}$  for fixed $\widehat{\bm{\alpha}}^{k+1}$ via coordinate descent algorithm.
       \item [(iii)] $\delta=\vert L_n(\widehat{\bm{\lambda}}^{k+1}, \widehat{f}_{\widehat{\bm{\lambda}}^{k+1}}) - L_n(\widehat{\bm{\lambda}}^{k}, \widehat{f}_{\widehat{\bm{\lambda}}^{k}})\vert$ and $e ={\Vert  \widehat{\bm{\lambda}}^{k+1} - \widehat{\bm{\lambda}}^{k} \Vert}_1$ .
      \end{itemize}   
    \Ensure  
      $\widehat{\bm{\alpha}}^{k+1}, \widehat{\bm{\lambda}}^{k+1}, \widehat{f}_{\widehat{\bm{\lambda}}^{k+1}}$.
  \end{algorithmic}  
\end{algorithm}

\noindent\textbf{Remark 1}. When updating $\bm{\alpha}$ interatively, for regression, it can be solved in a closed form as
$\widehat{ \bm{\alpha}}^{k+1} = (\Kb_{\widehat{\bm{\lambda}}^k, \sigma_n}^\intercal \Kb_{\widehat{\bm{\lambda}}^k, \sigma_n} +n\gamma_{1 n}\Kb_{\widehat{\bm{\lambda}}^k, \sigma_n})^{-1}\Kb_{\widehat{\bm{\lambda}}^k, \sigma_n}^\intercal Y$. For classification, we apply one-step Newton method for updating. Tuning parameters in the algorithm are chosen via cross-validation over a grid of $2^{-15}, 2^{-13}, \cdots, 2^{-13}, 2^{15}$. Although the kernel bandwidth, $\sigma_n$, can also be tuned, to save computation cost,  we follow \citet{jaakkola1999fisher} to set it to be the median value of the paired distances.

\section{Theoretical Properties}

In this section, we present some theoretical properties of our proposed method.  Since our proposed kernel function is new, we first provide two theorems that describe the properties for the RKHS generated by this kernel function. In the first theorem, we show that this space is dense in $L_2(P)$ subspace consisting of all measurable functions that only depend on the feature variables for which $\lambda_m\neq 0$ in the kernel function. In the second theorem, we obtain the entropy number for the unit ball in this space. Both theorems are necessary to establish the asymptotic properties of the proposed estimator for $f(\Xb)$ as given in the previous section.

To state our results, we define $f_0(\bX)$ as the Bayesian prediction function, which is assumed to be unique. That is, $E[l(Y, f)]$ attains its minimum when $f=f_0$. We assume that feature variables $X_1, X_2, \cdots, X_q$ are  important in terms that $f_0(\bX)$ is only a function of $X_1,X_2,...,X_q$ and for any $1\le s\le q$, 
$$E\left\{\Big(f_0(\bX)-E\left[f_0(\bX)\Big|X_1, X_2, X_{s-1}, X_{s+1} \cdots, X_q \right]\Big)^2\right\}> 0.$$
Finally, we let $d_2(f_0,  \mathcal{H}_{\bm{\lambda},\sigma_n})$ denote the $L_2(P)$-distance between $f_0$ and the RKHS generated by $\kappa_{\bm{\lambda},\sigma_n}$. 

\begin{theorem}For a vector $ \bm{\lambda}_n=(\lambda_{n1},...,\lambda_{np_n})$ with $\lambda_{nm} \geq 0$ for $m=1,..., p_n$, the following results hold:
\begin{enumerate}
\item[(i)] If $ \lambda_{nm} >  0$  for $m=1,...,q$, i.e., $\lambda_n$'s that are associated with the important  features are strictly positive, then
$d_2( f_0, \mathcal{H}_{\bm{\lambda}_n, {\sigma}_n}) \to 0.$
\item[(ii)] If for some  $m\le q$, $ \lambda_{nm} =  0$, then
$\lim\inf d_2( f_0, \mathcal{H}_{\bm{\lambda}_n,{\sigma}_n}) > 0.$
\end{enumerate}
\end{theorem}
Note: The Theorem holds for $\bm{\lambda}$ whose value depends on $n$ and denoted as $\bm{\lambda}_n$.
\begin{proof}
To prove (i), we first note that after expansion,
$\kappa_{\bm{\lambda}_n,  \sigma_n}(\Xb, \widetilde\Xb)$ is the summation of a number of Gaussian kernels. In particular, one term of this summation is
$$\left\{\lambda_{n1} \lambda_{n2} \cdots \lambda_{nq} {\kappa_{\sigma_{n}}(X_{1}, \widetilde X_{1})  \kappa_{\sigma_{n}}(X_{2}, \widetilde{X}_{2}) \cdots  \kappa_{\sigma_{n}}(X_{q}, \widetilde X_{q})}\right\},$$
where $\kappa_{\sigma}(x,y)=\exp\{-(x-y)^2/\sigma^2\}.$
Since  $\lambda_{n1},...,\lambda_{nq}>0$, the kernel function associated with this term is proportional to the Gaussian kernel in the space of $(X_{1}, \cdots, X_{q})$ with bandwidth $\sigma_{n}$ for each domain $k$. Therefore, the closure of the RKHS generated by $\kappa_{\bm{\lambda}_n, \sigma_n}$ includes the RKHS generated by
the Gaussian kernel in the space of $(X_{1}, \cdots, X_{q})$.
The result in (i) holds since the latter is asymptotically dense in the subspace of $L_2(P)$ consisting of any functions depending on $(x_1,...,x_q)$.

To prove (ii), if $\lambda_{m}=0$, then it is clear that any function in $\mathcal{H}_{\bm{\lambda}_n, {\sigma}_n}$ 
only depends on the feature variables except $X_{m}$.  Therefore, 
$$\mathcal{H}_{\bm{\lambda}_n,{\sigma}_n}\subset \left\{g(\bX_{-m}): g\in L_2(P)\right\},$$
where $\bX_{-m}$ denotes all the feature variables excluding $X_{m}$. 
On the other hand, the projection of $f_0$ on the latter space is $E[f_0|\bX_{-m}]$.
Therefore,
$$\lim\inf d(f_0, \mathcal{H}_{\bm{\lambda}_n,{\sigma}_n})\ge d(f_0,E[f_0|\bX_{-m}])>0$$
since $X_m$ is one important variable for $f_0$. We obtain the result.
\end{proof}

Our next theorem studies the bracket covering number for a unit ball in ${\cal H}_{\bm{\lambda}_n,\sigma_n}$. We consider ${\cal B}_n$ as the unit ball in ${\cal H}_{\bm{\lambda}_n,\sigma_n}$, i.e., 
${\cal B}_n\equiv \left\{f(\bx): \Vert f\Vert_{{\cal H}_{\bm{\lambda}_n,\sigma_n}}\le 1\right\},$
 Then the $\epsilon$-bracket covering number for ${\cal B}_n$, denoted as $N_{[]}(\epsilon, \mathcal{B}_n, \Vert\cdot\Vert_{L_2(P)})$,
is defined as the minimal number of pairs $[l(\bx), u(\bx)]$ such that any function $\Vert u(\Xb)-l(\Xb)\Vert_{L_2(P)}\le \epsilon$ and any function $f$ in ${\cal B}_n$ is between one pair, i.e., $l(\bx)\le f(\bx)\le u(\bx)$.
 
 \begin{theorem} For a vector $ \bm{\lambda}_n=(\lambda_{n1},...,\lambda_{np_n})$ such that $ \lambda_{nm}$ is uniformly bounded by a constant $M$ for $m=1,...,q$ and $\lambda_{n(q +1)}=...=\lambda_{np_n}=0$, it holds
$$ \log \mathcal{N}_{[]} (\epsilon, \mathcal{B}_n, \Vert \cdot \Vert_{L_2(P)} ) \leq C\sigma_n^{-(1-v/4)q} \epsilon^{-v}, $$
where $v$ is any constant within $(0,2)$ and $C$ only depends on $M$ and $q$.
\end{theorem}

\begin{proof}
For any $f\in {\cal B}_n$ with form
$$f(\bx)=\sum_{i=1}^{\infty} \alpha_i \kappa_{\bm{\lambda}_n,\sigma_n}(\bx, \bx_i),$$
where $\bx_1,\bx_2,...$ are a sequence of given points.
Using the expansion of $\kappa_{\bm{\lambda}_n,\sigma_n}$, we have
\begin{eqnarray*}
f(\bx)&=&\sum_{\{k_1,...,k_s\}\subset\{1,...,q\}\cup\phi} \lambda_{nk_1}\cdots\lambda_{nk_s}
\sum_{i=1}^{\infty} \alpha_i \exp\left\{{-\frac{(x_{ik_1}-x_{k_1})^2 + \cdots +(x_{ik_s}-x_{k_s})^2 }{\sigma_n^2}}\right\}\\
&=&\sum_{\{k_1,...,k_s\}\subset\{1,...,q\}\cup\phi}\sqrt{\lambda_{nk_1}\cdots\lambda_{nk_s}}
f_{k_1...k_s}(\xb),
\end{eqnarray*}
where $x_{ik}$ and $x_k$ are respectively the $k$th component of $\bx_i$ and $\bx$, and 
$$ f_{k_1...k_s}(\xb) = \sum_{i=1}^{\infty}  \alpha_i\sqrt{{\lambda}_{nk_1}\cdots{\lambda}_{nk_s}}\exp\left\{-\frac{(x_{ik_1}-x_{k_1})^2 + \cdots +(x_{ik_s}-x_{k_s})^2 }{\sigma_n^2}\right\}.$$ Here, if the index set if empty, then the exponential part in the summation is replaced by 1.

Clearly, if we denote ${\cal H}_{k_1...k_s}$ as the reproducing kernel Hilbert space generated by the Gaussian kernel
$\exp\left\{-{[(\widetilde x_{k_1}-x_{k_1})^2 + \cdots +(\widetilde x_{k_s}-x_{k_s})^2] }/{\sigma_n^2}\right\},$
then $f_{k_1...k_s}(\xb)\in {\cal H}_{k_1...k_s}$ and moreover,
\begin{align*}
\Vert f \Vert _{\mathcal{H}_{{\bm{\lambda}_n}, \sigma_n}}^2 &= \sum_{i=1}^{\infty} \sum_{j=1}^{\infty} \alpha_i\alpha_j \kappa_{{\bm{\lambda}_n, \sigma_n}}(\xb_i, \xb_j)\\
&= \sum_{i=1}^n \sum_{j=1}^n \alpha_i\alpha_j \sum_{\{k_1,...,k_s\}\subset\{1,...,q\}\cup\phi}{\lambda}_{nk_1}\cdots{\lambda}_{nk_s}\exp\left\{-\frac{(x_{ik_1}-x_{jk_1})^2 + \cdots +(x_{ik_s}-x_{jk_s})^2 }{\sigma_n^2}\right\}\\
&=  \sum_{\{k_1,...,k_s\}\subset\{1,...,q\}\cup\phi}\sum_{i=1}^{\infty} \sum_{j=1}^{\infty} \alpha_i\alpha_j {\lambda}_{nk_1}\cdots{\lambda}_{nk_s}\exp\left\{-\frac{(x_{ik_1}-x_{jk_1})^2 + \cdots +(x_{ik_s}-x_{jk_s})^2 }{\sigma_n^2}\right\}\\
&=  \sum_{\{k_1,...,k_s\}\subset\{1,...,q\}\cup\phi} \Vert f_{k_1...k_s} \Vert _{\mathcal{H}_{k_1...k_s}}^2.
\end{align*} 
Thus, $\Vert f\Vert_{\mathcal{H}_{\bm{\lambda}_n,\sigma_n}}\le 1$ implies $\Vert f_{k_1...k_s} \Vert _{\mathcal{H}_{k_1...k_s}}\le 1$ for any $k_1,...,k_s$.

Consequently, since such $f$ is dense in $\mathcal{B}_n$, we conclude
$$\mathcal{B}_n \subseteq \overline{\left\{ \sum_{\{k_1,...,k_s\}\subset\{1,...,q\}\cup\phi} f_{k_1...k_s}(\bx) \sqrt{{\lambda}_{nk_1}\cdots{\lambda}_{nk_s}}: \Vert f_{k_1...k_s} \Vert _{\mathcal{H}_{k_1...k_s}}^2 \leq 1\right \}}.$$
Thus, there exists a constant $C$ only depending on $M$ and $q$ such that
$$\log\mathcal{N}_{[]}(2^qM^{q/2}\epsilon, \mathcal{B}_n,\Vert \cdot\Vert _{L_2(P)})  \leq  {\sum_{\{k_1,...,k_s\}\subset\{1,...,q\}\cup\phi}} \log\mathcal{N}_{[]}(\epsilon, \{f_{k_1...k_s}(\bx), \Vert f_{k_1...k_s} \Vert _{\mathcal{H}_{k_1...k_s}} \leq  1 \}, \Vert \cdot\Vert _{L_2(P)})  $$
According to (\citet{steinwart2007gaussian}), we know
$$\log\mathcal{N}_{[]}(\epsilon, \{f_{k_1...k_s}(\bx), \Vert f_{k_1...k_s} \Vert _{\mathcal{H}_{k_1...k_s}}^2 \leq  1 \}, \Vert \cdot\Vert _{L_2(P)}) \leq C\sigma_n^{-(1-v/4)s} \epsilon^{-v}, $$
for any constant $v\in (0,2)$ and a constant $C$ only depending on $s$. Therefore,
$$ \log \mathcal{N} (\epsilon, \mathcal{B}_n, \Vert \cdot \Vert_{L_2(P)} ) \leq C(M,q)\sum_{\{k_1,...,k_s\}\subset\{1,...,q\}\cup\phi}\sigma_n^{-(1-v/4)s} \epsilon^{-v}\le C(M, q) \sigma_n^{-(1-v/4)q}\epsilon^{-v}$$
for a constant $C(M, q)$.
We have proved Theorem 2.
\end{proof}

Our next theorem gives the main properties of the estimated prediction function. We show that the resulting prediction function from our method leads to Bayesian risk asymptotically. Moreover, with probability tending to one, the variable selection based on non-zero $\lambda_n$'s  is oracle  as if we knew which variables were important. 
 Recall that $(\widehat{\bm{\lambda}_n}, \widehat f)$ is the optimal solution of the objective function 
\begin{equation}
L_n(\bm{\lambda}_n, f) = {\bf P}_n \mathnormal{l}(Y, f(\Xb)) +\gamma_{1n}\Vert f\Vert _{\mathcal{H}_{\bm{\lambda}_n,  \bm{\sigma}_n}}^2 +\gamma_{2n} P(\bm{\lambda}_n),
\end{equation}
where $P(\bm{\lambda}_n)$ is the truncated Lasso penalty for $\bm{\lambda}_n$. 
Equivalently, if we define for any $\bm{\lambda}_n$, 
$$\widehat{f}_{\bm{\lambda}_n} = \textrm{arg}\min_{f}L_n(\bm{\lambda}_n,f),$$
which exists due to the convexity of $L_n(\bm{\lambda}_n, f)$ in $f$,
then 
$\widehat{\bm{\lambda}}$ minimizes $L_n(\widehat{\bm{\lambda}}_n,\widehat f_{\bm{\lambda}_n})$
and $\widehat f=\widehat f_{\widehat{\bm{\lambda}}_n}$.

For the main theorem, we assume $(Y,\Xb)$ to have a bounded support and need the following conditions. 
 \\
 (C1).  The loss function $l(y,f)$ is convex  and is Lipschtisz continuous with respect to $f$ in any bounded set.\\
(C2). There exit $\delta>0$ and a constant $c_1>0$ such that
$$E[l(Y,f(\bX))-l(Y,f_0(\bX))]\ge c_1 \Vert f(\bX)-f_0(\bX)\Vert_{L_2(P)}^{2}$$ whenever $E[l(Y,f(\bX))-l(Y,f_0(\bX))]$ is smaller than $\delta$.
\\
(C3).
Assume $\Vert l_2(Y,f(\Xb))-l_2(Y,f_0(\Xb))\Vert_{L_2(P)} \le c_2 \Vert f(\bX)-f_0(\bX)\Vert_{L_2(P)}$ for a constant $c_2$, where $l_2(y,x)=\partial l(y,x)/\partial x$.
\\
(C4). For any $\widetilde{\bm{\lambda}_n}=(\lambda_{n1},...,\lambda_{np_n})$ such that $\lambda_{nk}=0$ for $k>q$, let $\Lambda_{\max}(\Xb_{-q})$ and $\Lambda_{\min}(\Xb_{-q})$ be the largest and smallest eigenvalues of the matrix $\left(E[K_{\widetilde{\bm{\lambda}}_n}(\Xb_j, \Xb) K_{\widetilde{\bm{\lambda}}_n}(\Xb_l, \Xb)|\Xb_{-q}]\right)$ where $\Xb_{-q}$ denotes all unimportant variables. We assume that with probability one,
there exists one constant $c$ such that $\Lambda_{\max}(\Xb_{-q})/\Lambda_{\min}(\Xb_{-q})\le c\sigma_n^{-1/2}$ and
$E[\Lambda_{\min}(\Xb_{-q})\kappa_n(x,X_m)^2]\le c\sigma_n^{1/2}$ for any $m>q$.
\\
(C5).  Assume $\log p_n=o(n^{1-(2+q)\alpha_1 - \alpha_2 - \alpha_3})$. Moreover, we assume $\sigma_n = n^{-\alpha_1}$, $\gamma_{1n} = n^{-\alpha_2}, \gamma_{2n} = n^{-\alpha_3}$, where $\alpha_k>0$ for $k=1,2,3$ and they satisfy

(i) $ 1- (2+q) \alpha_1 - \alpha_2 >0$ 

(ii) $0 < \alpha_3 <\min{ \Big( \frac{1}{4} (1+ \frac{\alpha_1 q}{2} + \alpha_2), 1- (2+q)\alpha_1 - \alpha_2, \frac{\alpha_1}{2}, \frac{\alpha_2}{2} \Big)} $.

Conditions (C1)-(C3) give the assumptions for the loss functions. It can be verified that they hold for $l(y,f)=(y-f)^2$ for a continuous $y$ and for $l(y,f)=\exp(-yf)$ for a binary $y$.  Condition (C4) 
implies the equivalence between the Euclidean norm of the coefficients and the reproducing kernel Hilbert space norm, up to a scale proportional to $\sigma_n^{-1/2}$. The second half of the condition in (C4) holds automatically if the important variables are independent of the unimportant variable when $\Lambda_{\min}(\Xb_{-q})$ does not depend on $\Xb_{-q}$.  We note that such a condition is analogue to the design matrix condition assumed in  high dimensional linear model literature.
Finally, condition (C5) allows the dimensionality of the feature variable to be ultra-high and imposes additional constraints for the choices of the bandwidth and two tuning parameters.
 
\begin{theorem} Under Conditions (C1)-(C5), there exists a local minimizer $\widehat{\bm{\lambda}}_n$ for $L_n(\bm{\lambda}_n,\widehat f_{\bm{\lambda}_n})$  such that with probability tending to one,\\
(a) $E[l \big( Y, \widehat{f}_{\widehat{\bm{\lambda}}_n}  \big)]$ converges to $E[l\big( Y, f_0\big)]$.\\
(b)  For $m = 1,...,q,  \widehat{\lambda}_{nm} > 0$.\\
(c) For $ m = q+1, q+2, \cdots, p_n, \  \widehat{\lambda}_{nm} = 0$.
\end{theorem}

The first part of Theorem 3 implies that the loss of the estimated prediction function converges to the Bayes risk. The last two conclusions in  Theorem 3 show that the $\widehat{\lambda}_{nm}$'s associated with important feature variables should be non-zero, i.e., the estimated function does depend on important variables. More importantly, the proposed method can estimate the predicted function as if we knew which variables are important in the truth. 
The proof for Theorem 3 is given in the supplementary file. The proof of Theorem 3(a) entails careful examination of the stochastic variability of $L_n(\bm{\lambda}_n, \widehat f_{\bm{\lambda}_n})$, for which we first establish a preliminary bound for $\widehat f_{\bm{\lambda}_n}$ and then appeal to some concentration inequalities for empirical processes with metric entropy as derived from Theorem 2. To prove Theorem 3(b) and (c) in the theorem, we examine the KKT conditions to show that the oracle estimators, i.e., $\lambda_{nm}$ is known to be zero for $m>q$, satisfies the KKT conditions with probability tending to one. Again, concentration inequalities for empirical processes are needed in technical arguments in the proof.

\section{Simulation Study}

We conducted two simulation studies, one for a regression problem with continuous $Y$ and the other for classification with binary $Y$.  In the first simulation study, we considered a continuous outcome model with total number of $p$ correlated feature variables, which were generated from a multivariate normal distribution, each with mean zero and variance one. Furthermore, $X_{1}, X_{2},X_{3}, X_{4}$ were correlated with  $corr(X_{1}, X_{2}) = 0.4,\ corr(X_{1},X_{3}) = -0.3, \ corr(X_{2}, X_{3}) = 0.5$ and $corr(X_{3}, X_{4}) = 0.2$, while the others were all independent.
The outcome variable, $Y$, was simulated from a linear model 
$$Y= 0.9 X_{5}^3 +  4  X_{1}X_{2}X_{3}+ 2.3\exp(-X_{3}) + 4 X_{4} + \epsilon,$$
where $\epsilon \sim N(0,1)$. Thus, 
 $X_{1}$ to $X_5$ were important variables but not any others.
  In the second simulation study,  $X$'s were generated similarly but with some different correlations:  $corr(X_{1}, X_{2}) = -0.2,\ corr(X_{1},X_{4}) = 0.2, \ corr(X_{2}, X_{3}) = 0.5, \ corr(X_{3}, X_{4}) = 0.3$ and $corr(X_{3}, X_{4}) = -0.4$. The binary outcome, $Y$, with values $-1$ and $1$, were generated from a Bernoulli distribution with the probability of being one given by 
$$ \Big\{1+ e^{-0.25 +( X_{2} - 1.1X_{3} +0.3 X_{4})^3  } \Big\}^{-1},$$
so only $X_2$ to $X_4$ were important variables.
Since many biomedical applications (as well as our application in this work) have small to moderate sample sizes, in both simulation studies, we considered sample size $n=100, 200$ and $400$ and varied the feature dimension from  $p = 200 , 400$ to  $1000$. Each simulation setting was repeated 500 times.

For each simulated data, we used the proposed method to learn the prediction function. Initial values, tuning parameters  and optimization package used for binary case are chosen as in Remark 1 of Section 2, where 3-fold cross-validation was used for selecting the tuning parameters. The bound of regularized parameter $M$ was chosen to be $10^5$.  We also centerized continuous outcome and re-weighted class label controlled to be balanced before iteration to make numerical stable. We reported the true positive rates, true negative rates and the average number of the selected variables for feature selection. We also reported the prediction errors or misclassification rates using a large and independent validation data.
For comparison, we compared our proposed method with HSICLasso and SpAM since both methods were able to estimate nonlinear functions in high dimensional settings.
In addition, we also compared the performance with LASSO in the first simulation study and  $l_1$-SVM in the second simulation study, in order to study the impact due to model misspecification. 

The results based on 500 replicates are summarized in Tables 1 and 2.   From these tables, we observe that for fixed dimension, the performance of our method improves as sample size $n$ becomes large in terms of the improved true positive and true negative rates for feature selection as well as decreasing prediction errors. In almost all cases, our true negative rate is close to 100$\%$, which shows that noise variables can be identified with a very high chance. As expected, the performance deteriorates as the dimensionality increases. Interestingly, our method continues to select only a small number of feature variables.
Comparatively, HSICLasso selected many more noise variables and had larger prediction errors, while SpAM also tended to select more features than our method. The performance of these methods become much worse when the feature dimension is 1000.
Clearly, LASSO and $l_1$-SVM did not yield reasonable variable selection results and their prediction errors are much higher due to model misspecification.
We also give boxplots to visualize prediction performance of 500 replications in Figures 1 and 2. Since Lasso cannot provide stable prediction errors, its prediction errors from many replicates are out of the bound as shown in Figure 1. Figure 1 and 2 further confirm that our method is superior to all other methods, even when the dimension is as large as 1000 and the sample size is as small as $n=100$, which is of similar size as our real data analysis example in Section 5.

\begin{table}
\caption{Results from The Simulation Study with Continuous Outcome}

\begin{center}
\scalebox{0.75}{
\begin{tabular}{ccccccccccccccccc} 

{} & \multicolumn{16}{c}{(a) Summary of Feature Selection Performance}  \\
\hline
{} & {} & \multicolumn{3}{c}{Proposed Method}  & &  \multicolumn{3}{c}{HSICLasso} & &  \multicolumn{3}{c}{SPAM} & &  \multicolumn{3}{c}{LASSO}  \\
\cline{3-5}\cline{7-9}\cline{11-13}\cline{15-17} 
$p$  & $n$   & TPR   &  TNR & Avg$\#$  & & TPR & TNR & Avg$\#$   & & TPR & TNR & Avg$\#$   & & TPR & TNR & Avg$\#$ \\ 
\hline
100 & 100 & 60.9$\%$ & 97.3$\%$ & 5.6 & & 81.5$\%$ &78.5$\% $ &24.5 & & 99.6$\%$ &34.6$\% $ &67.1 & & 98.8$\%$ &1.3$\% $ &98.8 \\
   & 200 & 71.2$\%$   & 99.0$\%$ & 4.5 & & 98.0$\%$ & 60.4 $\%$ & 42.5 & & 100.0$\%$ &4.4$\% $ &95.8  & & 100.0$\%$ &0.1$\% $ &99.9 \\
   & 400 & 82.7$\%$    & 98.4$\%$ & 5.7 & & 99.6 $\%$ & 78.0 $\%$  & 25.8 & & 100.0$\%$ &0.3$\% $ &99.7 & & 100.0$\%$ &0.1$\% $ &99.9 \\ 
\hline
200 & 100 & 57.2$\%$ & 98.7$\%$ & 5.5  & & 75.6$\%$  & 88.8$\%$ &25.6   & & 99.1$\%$  & 63.0$\%$ &77.1 & & 84.0$\%$ &52.2$\% $ &97.5 \\
 & 200 & 66.6$\%$ & 99.5$\%$ & 4.2  & & 94.0$\%$ & 75.2$\%$ & 53.1 & & 100.0$\%$ & 33.7$\%$ & 134.1 & & 99.1$\%$ &0.0$\% $ &198.6 \\
   & 400 & 78.1$\%$   & 99.4$\%$ & 5.0 & & 99.8 $\%$ & 84.2$\%$ & 35.8 & & 100.0$\%$ & 5.4$\%$ & 189.5  & & 100.0$\%$ &0.12$\% $ &199.8  \\ 
\hline
400 & 100 & 47.3$\%$ & 99.3$\%$ & 5.2 & & 68.5$\%$ &90.4$\%$  &41.5 & & 98.2$\%$  & 80.8$\%$ &80.8 & & 79.4$\%$ &76.4$\% $ &97.1  \\
 & 200 & 65.0$\%$ & 99.7$\%$ & 4.5 & & 86.3$\%$ & 89.0$\%$ &47.6 & & 100.0$\%$  & 62.1$\%$ &154.6 & & 90.7$\%$ &51.4$\% $ &196.6 \\
 & 400 & 73.1$\%$ & 99.8$\%$ & 4.4 & & 99.7$\%$ & 87.6$\%$ &54.0 & & 100.0$\%$  & 34.0$\%$ &265.8  & & 99.1$\%$ &0.7$\% $ &397.3 \\
\hline
1000 & 100 & 40.7$\%$ & 99.7$\%$ & 5.0  & & 56.0$\%$ &91.8$\%$  &84.5 & & 93.7$\%$ &92.2$\%$  &82.6  & & 73.6$\%$ &90.6$\% $ &97.2 \\
 & 200 & 61.2$\%$ & 99.9$\%$ & 4.5  & & 78.2$\%$ & 98.6$\%$ &18.4 & & 99.9$\%$ &84.2$\%$  &162.0 & & 85.5$\%$ &80.7$\% $ &196.1 \\
  & 400 & 70.7$\%$ & 99.9$\%$ & 4.0 & & 99.4$\%$ & 91.0$\%$ &94.9 & & 100.0$\%$ &68.7$\%$  &316.4 & & 94.5$\%$ &60.7$\% $ &395.6 \\
  \hline   
\end{tabular}
}
\end{center}
\begin{center}
\scalebox{0.75}{
\begin{tabular}{ccccccc}
\\
 &  &  & \multicolumn{4}{c}{(b) Summary of Prediction Errors}\\ \hline
 $p$  & $n$  & & Proposed Method   &HSICLasso &SPAM &LASSO\\
 \hline
100 & 100 & & 7.405 (0.527)  &7.695 (0.291) & 6.985 (0.291) & 41.663 (11.325)\\
& 200& &  5.929 (0.950)  & 7.323 (0.098) & 7.299 (0.389) & 9.508 (0.575)\\
& 400& & 4.424 (0.777)   &7.115 (0.053)   &6.868 (0.292) & 7.840 (0.183)\\
\hline
200 & 100 & & 7.567 (0.493)  & 7.603 (0.286)  & 6.672 (0.336) &10.176 (0.794)\\
& 200 & & 6.623 (0.412) & 7.313 (0.130) & 6.404 (0.279) &44.464 (9.125)\\
&400 & & 5.661(0.580)  & 6.946 (0.054) & 6.767 (0.305) &   9.370 (0.433)\\
\hline
400 & 100 & & 7.920 (0.670) & 8.001 (0.284) & 6.815 (0.399) &  9.091 (0.532)\\
& 200 & & 7.008 (0.346) & 7.563 (0.233)  &6.222 (0.218) &  10.151 (0.722)\\
& 200 & & 6.444 (0.199) & 7.061 (0.049)  & 6.079 (0.192) &  40.190 (6.402)\\
\hline
1000 & 100 & & 8.215 (0.764) & 8.638 (0.263) & 7.067 (0.372) &  8.851 (0.406) \\
&200 & & 7.324 (0.368) & 7.539 (0.252)  & 6.214 (0.242) &  8.871 (0.379)\\
&400 & & 6.818 (0.250) & 7.376 (0.068)  & 5.870 (0.161) &  9.652 (0.429)\\
\hline
\end{tabular}
}
\end{center} 
{Note. In (a),  ``TPR" is the true positive rate,  ``TNR" is the true negative rate, and  ``Avg$\#$" is the average number of the selected variables from 500 replicates. In (b), the numbers are the mean squared errors from prediction, and the numbers within parentheses are the median absolute deviations from 500 replicates.}
\end{table}

\begin{table}
\caption{Results from The Simulation Study with Binary Outcome}

\begin{center}
\scalebox{0.75}{
\begin{tabular}{ccccccccccccccccc} 

{} & \multicolumn{16}{c}{(a) Summary of Feature Selection Performance}  \\
\hline
{} & {} & \multicolumn{3}{c}{Proposed Method}  & &  \multicolumn{3}{c}{HSICLasso} & &  \multicolumn{3}{c}{SPAM} & &  \multicolumn{3}{c}{$l_1$-SVM}  \\
\cline{3-5}\cline{7-9}\cline{11-13}\cline{15-17} 
$p$  & $n$   & TPR   &  TNR & Avg$\#$  & & TPR & TNR & Avg$\#$   & & TPR & TNR & Avg$\#$   & & TPR & TNR & Avg$\#$ \\ 
\hline
100 & 100 & 74.7$\%$ & 99.0$\%$ & 3.3 & & 71.1$\%$ &79.4$\% $ &22.1  & & 64.5$\%$ &89.6$\% $ &12.1 & &76.2$\%$  & 75.1$\%$ & 26.5 \\
   & 200 & 83.9$\%$   & 99.9$\%$ & 2.6 & & 80.7$\%$ & 89.7 $\%$ & 12.4  & & 53.4$\%$ &98.9$\% $ &2.6  & & 92.5$\%$ &80.4$\% $ &21.8 \\
   & 400 &  86.0$\%$ &  99.9$\%$& 2.6  & & 87.8$\%$   & 90.3$\%$ & 12.1  & & 50.6$\%$ &99.9$\% $ &1.5  & & 98.8$\%$ &71.3$\% $ &30.8\\
   \hline
   200 & 100 & 70.4$\%$ & 99.3$\%$ & 3.5 & & 71.3$\%$ &80.1$\% $ &41.3  & & 63.6$\%$ &91.2$\% $ &19.2 & &  71.3$\%$ & 85.4$\%$  &31.0  \\
   & 200 & 84.1$\%$   & 99.8$\%$ & 2.9 & & 78.3$\%$ & 95.0 $\%$ & 12.2  & & 54.5$\%$ &98.7$\% $ &4.1  & & 90.7$\%$ &80.2$\% $ &41.8 \\
   & 400 & 87.0$\%$    & 100.0$\%$ & 2.7 & & 83.1 $\%$ & 96.5$\%$  & 9.3  & & 50.6$\%$ &99.9$\% $ &1.6 & & 89.3$\%$ &73.4$\% $ &55.0 \\ 
   \hline
   400 & 100 & 68.5$\%$ & 99.5$\%$ & 3.9 & & 70.9$\%$ &79.4$\% $ &84.0  & & 63.7$\%$ &92.8$\% $ &30.6 & &  65.9$\%$ & 86.7$\%$  &54.7  \\
   & 200 & 84.5$\%$   & 99.9$\%$ & 3.0 & & 76.9$\%$ & 95.5 $\%$ & 20.2 & & 57.7$\%$ &98.3$\% $ &8.3  & & 87.0$\%$ &91.0$\% $ &38.1  \\
   & 400 & 87.0$\%$    & 100.0$\%$ & 2.6 & & 79.6 $\%$ & 98.9$\%$  &6.8 & & 51.9$\%$ &100.0$\% $ &1.8 & & 99.1$\%$ &82.3$\% $ &73.0 \\ 
   \hline
   1000 & 100 & 61.3$\%$ & 99.8$\%$ & 4.1 & & 72.2$\%$ & 77.4$\%$ & 227.4 & & 61.0$\%$  & 95.7 $\%$ & 45.0  & & 58.4$\%$  &  90.3 $\%$ & 98.9   \\
   & 200 & 86.3$\%$   & 99.9$\%$ & 3.3 & & 75.5 $\%$ & 95.9 $\%$ & 43.6  & & 54.3$\%$  &  98.9 $\%$ &  12.8  & &  79.4$\%$ & 91.4$\%$& 87.4 \\
   & 400 & 87.7$\%$    & 100.0$\%$ & 2.8 & & 73.9 $\%$ &  99.6 $\%$ & 6.6 & &50.0$\%$   & 100.0$\%$ & 1.9  & &  96.8$\%$ & 90.1 $\%$ & 101.6   \\ 
   \hline
\end{tabular}
}
\end{center}
\begin{center}
\scalebox{0.75}{
\begin{tabular}{ccccccc}
\\
 &  &  & \multicolumn{4}{c}{(b) Summary of Misclassification Errors}\\ \hline
 $p$  & $n$  & & Proposed Method  &HSICLasso &SPAM &$l_1$-SVM\\
 \hline
100 & 100 & & 0.314 (0.017)  &0.345 (0.028) & 0.343 (0.018)  &  0.359 (0.032)\\
& 200& &  0.290 (0.009)  & 0.307 (0.012) & 0.312 (0.002) & 0.305 (0.011)\\
& 400& & 0.283 (0.004)   & 0.292 (0.012) & 0.297 (0.002) & 0.292 (0.007) \\
\hline
200 & 100 & & 0.316 (0.019)  &0.351 (0.042) & 0.344 (0.034)  & 0.352 (0.031)\\
& 200 & &  0.280 (0.008)  & 0.302 (0.015) & 0.302 (0.003) & 0.321 (0.028)\\
&400 & & 0.270 (0.004)   & 0.282 (0.014)   &0.297 (0.002) & 0.326 (0.025)\\
\hline
400 & 100 & & 0.331 (0.024)  &0.372 (0.047) & 0.369 (0.046)  & 0.390 (0.031)\\
& 200 & &  0.286(0.010)  & 0.319 (0.018) & 0.311 (0.003) & 0.327 (0.026)\\
& 200 & & 0.277 (0.004)   & 0.288 (0.014)   &0.305 (0.001) & 0.295 (0.010)\\
\hline
1000 & 100 & & 0.352 (0.027)  & 0.397 (0.037) &  0.390 (0.036)  &0.416 (0.027)  \\
&200 & &  0.287 (0.008)  & 0.335 (0.024) & 0.315 (0.003) & 0.381 (0.020)\\
&400 & & 0.277 (0.004)   &  0.294 (0.008)  & 0.305 (0.001) & 0.353 (0.016)\\
\hline
\end{tabular}
}
\end{center} 
{Note. See Table 1.}
\end{table}

\begin{figure}
\caption{Boxplots of Prediction Errors for  Continuous Outcome }
\begin{center}
\centering
\includegraphics[width=3.3in, height=2.5in]{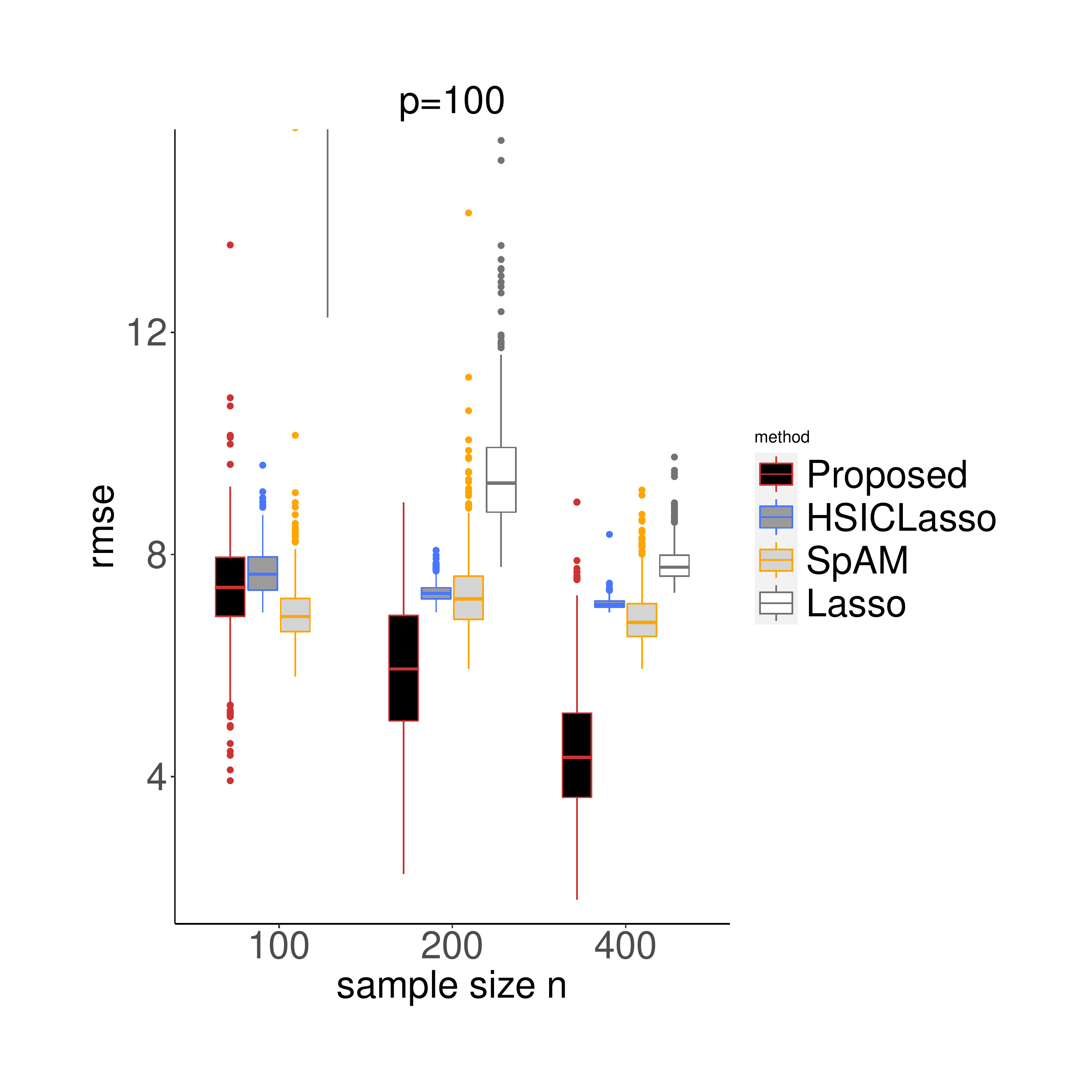}
\centering
\includegraphics[width=3.3in, height=2.5in]{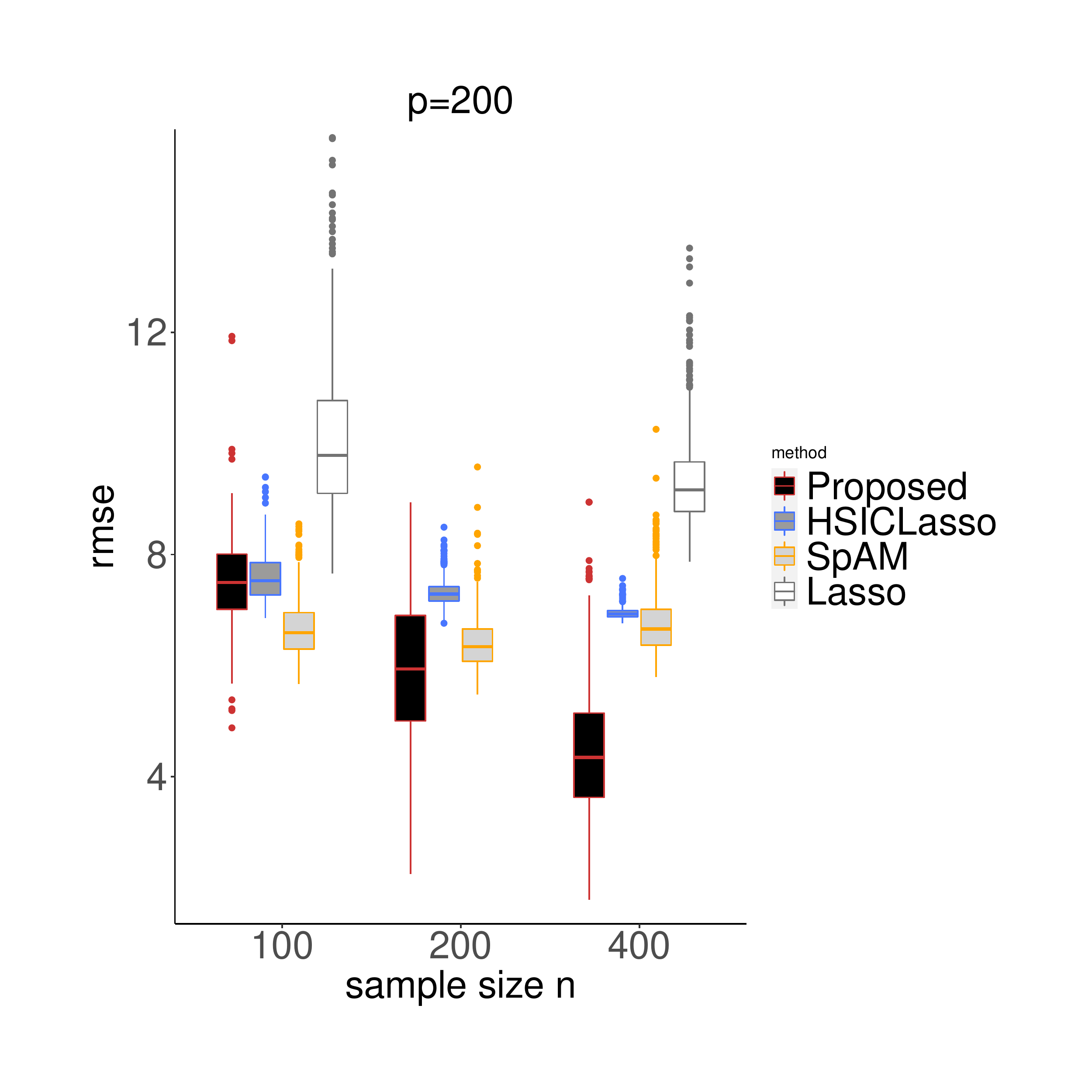}
\centering
\includegraphics[width=3.3in, height=2.5in]{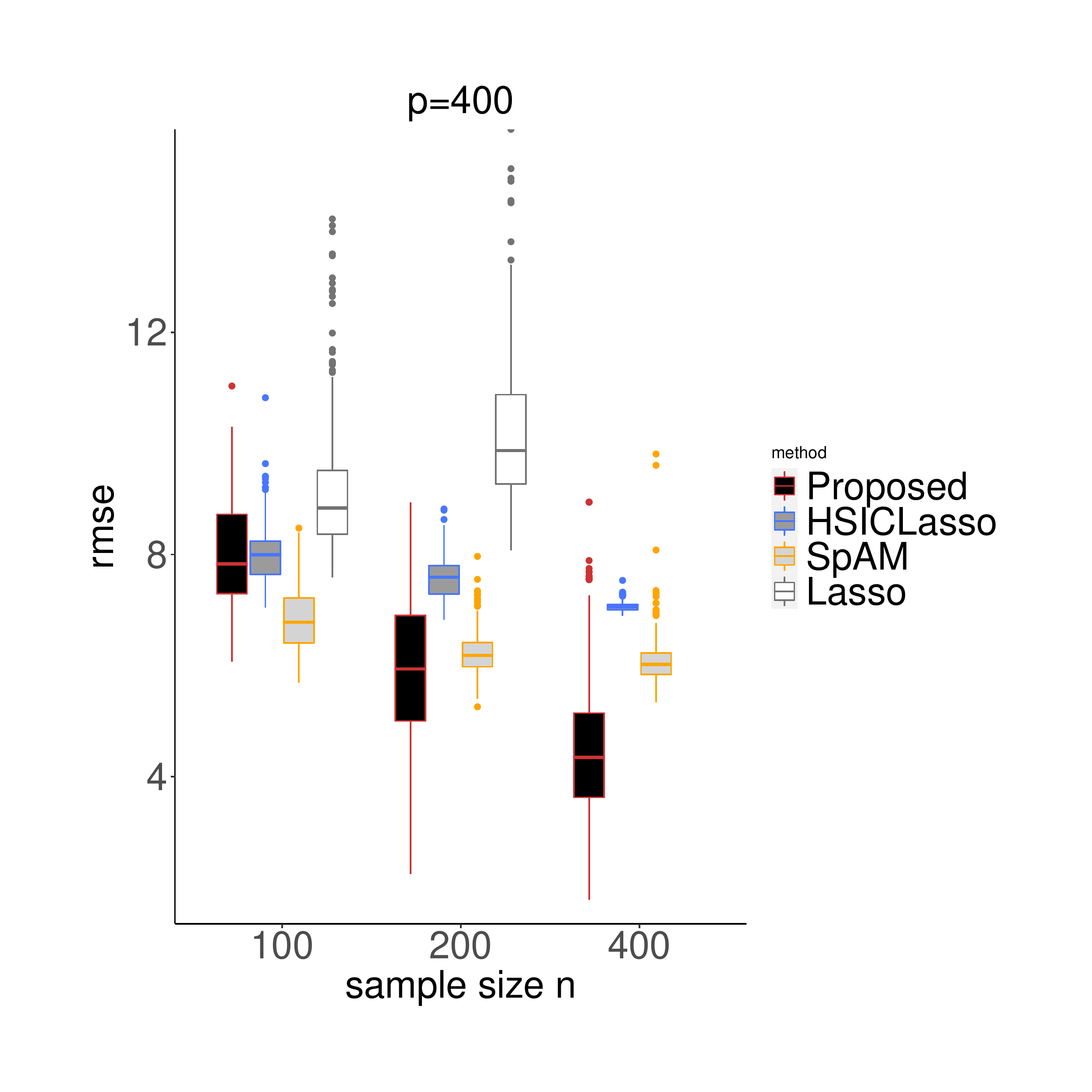}
\centering
\includegraphics[width=3.3in, height=2.5in]{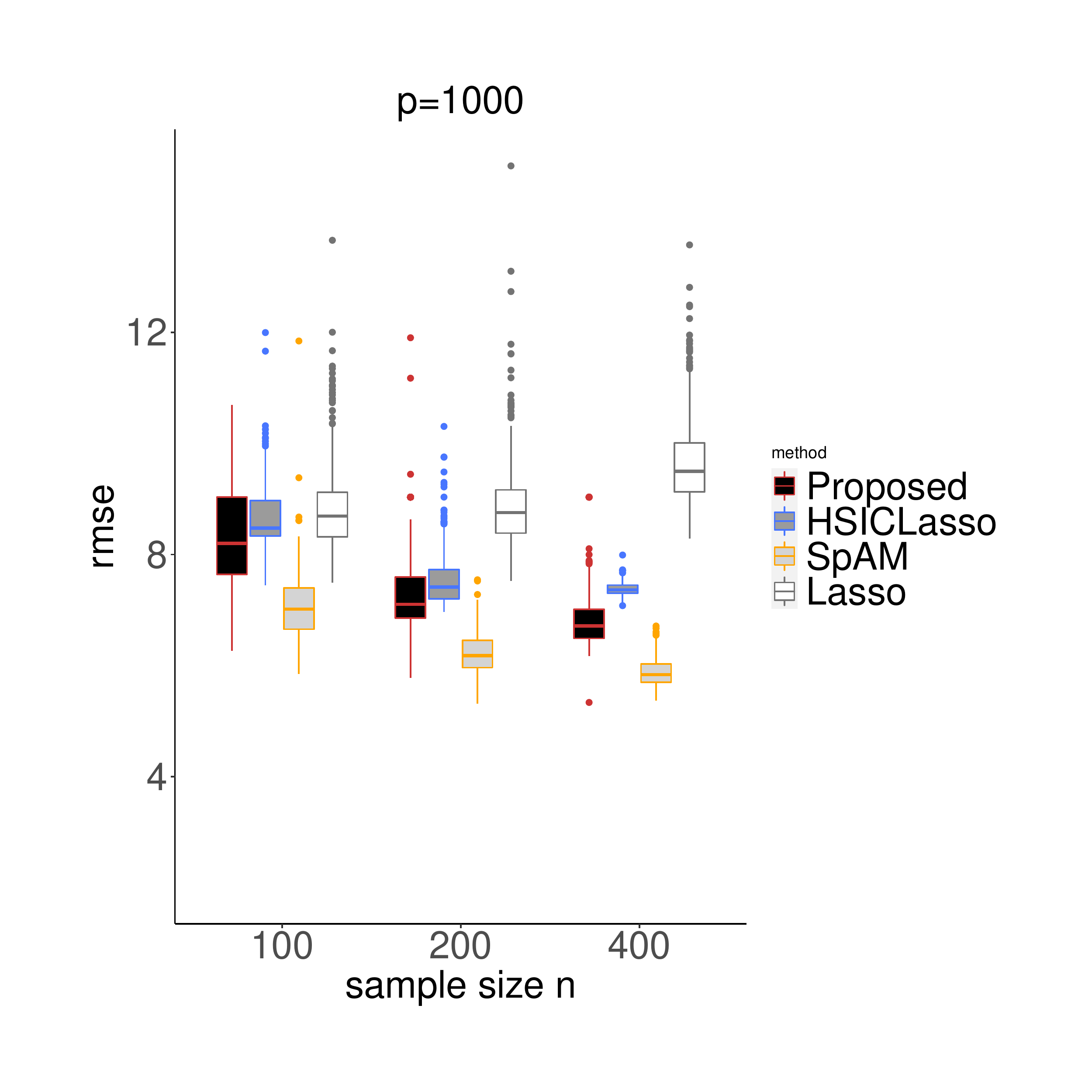}
\end{center}
{Note. The plots give the distribution of prediction errors among four competing methods. The comparing methods from left to right in each plot are our proposed method, HSICLasso, SpAM and Lasso. }
\end{figure}

\begin{figure}
\caption{Boxplots of Misclassification Errors for Binary Outcome }

\begin{center}
\includegraphics[width=3.3in, height=2.5in]{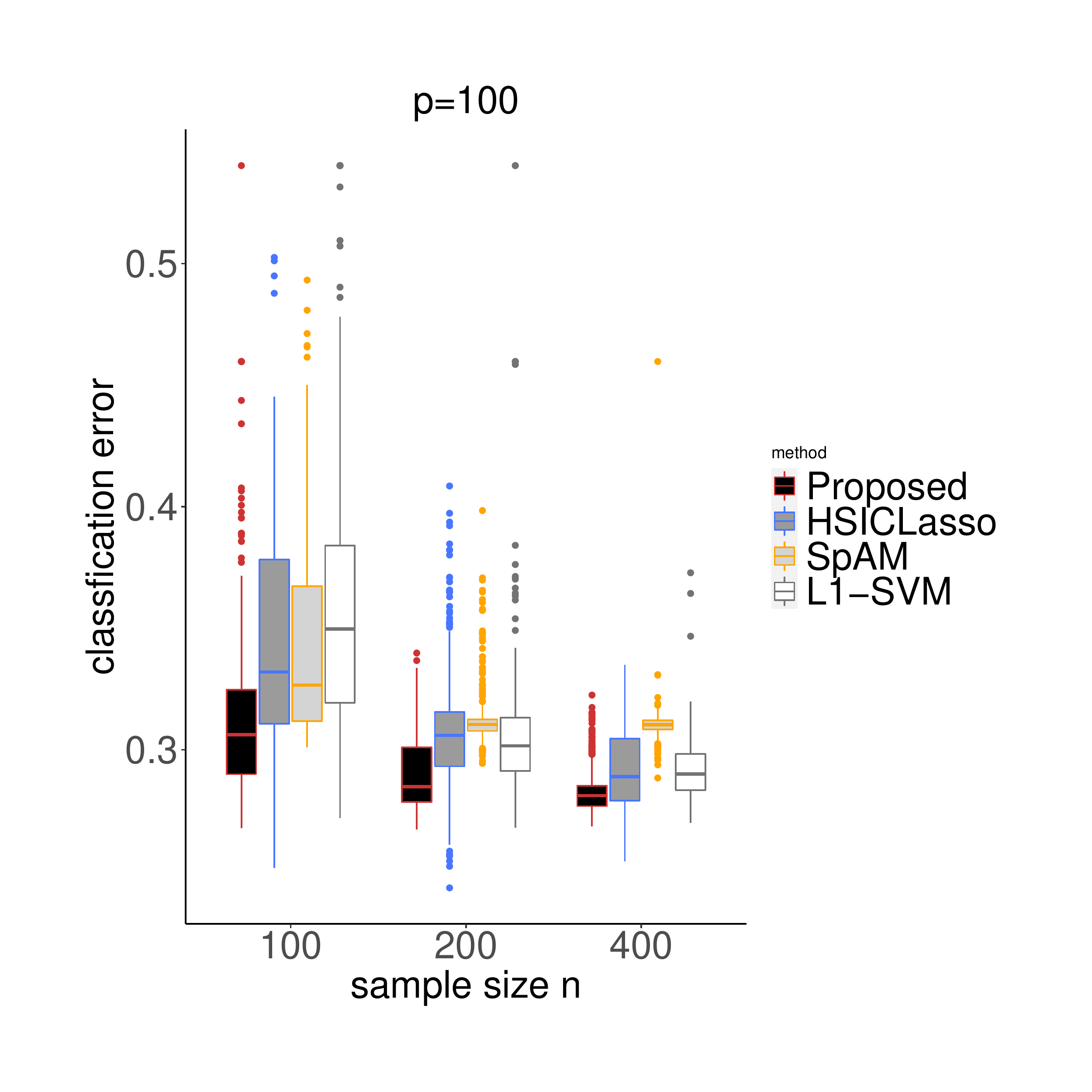}
\includegraphics[width=3.3in, height=2.5in]{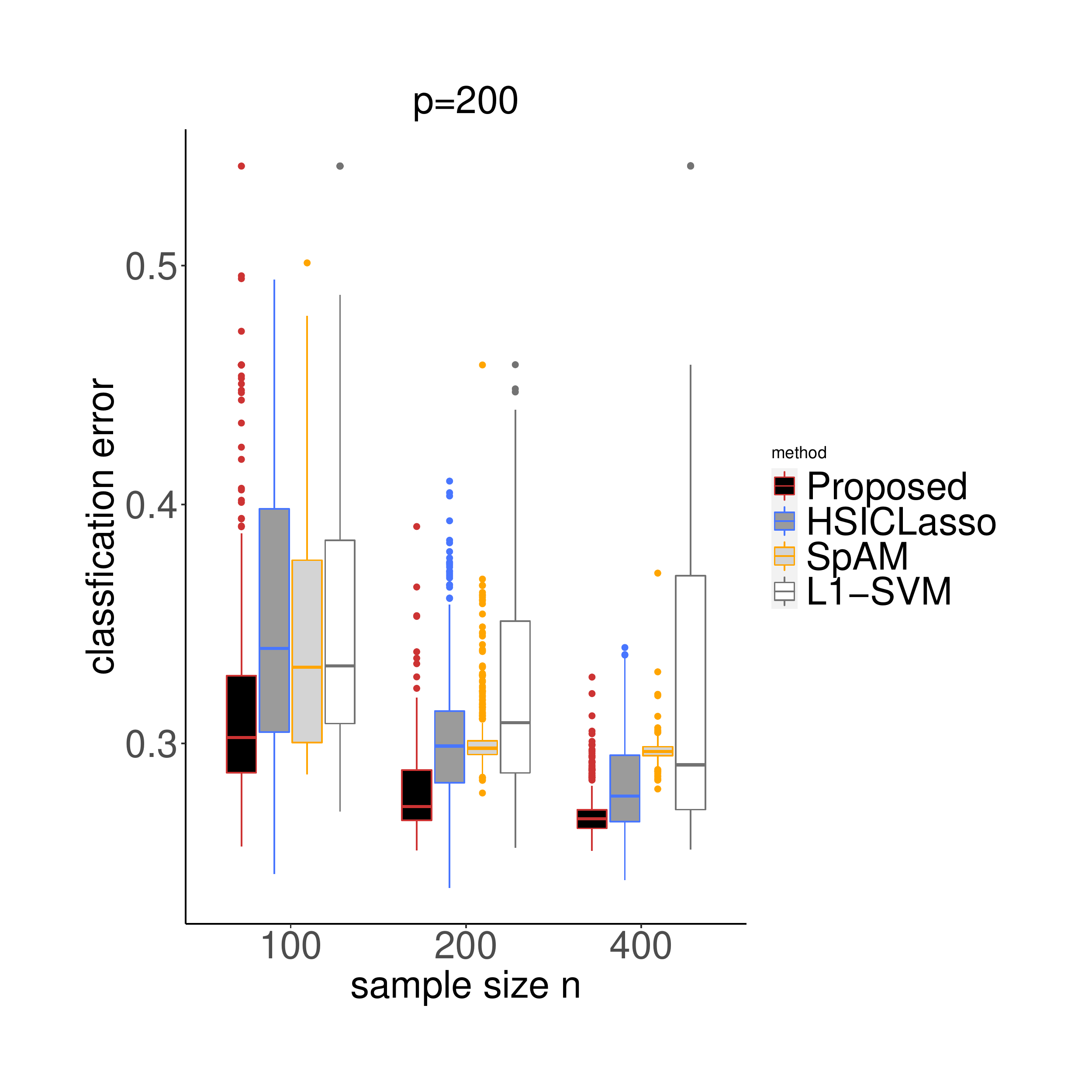}
\includegraphics[width=3.3in, height=2.5in]{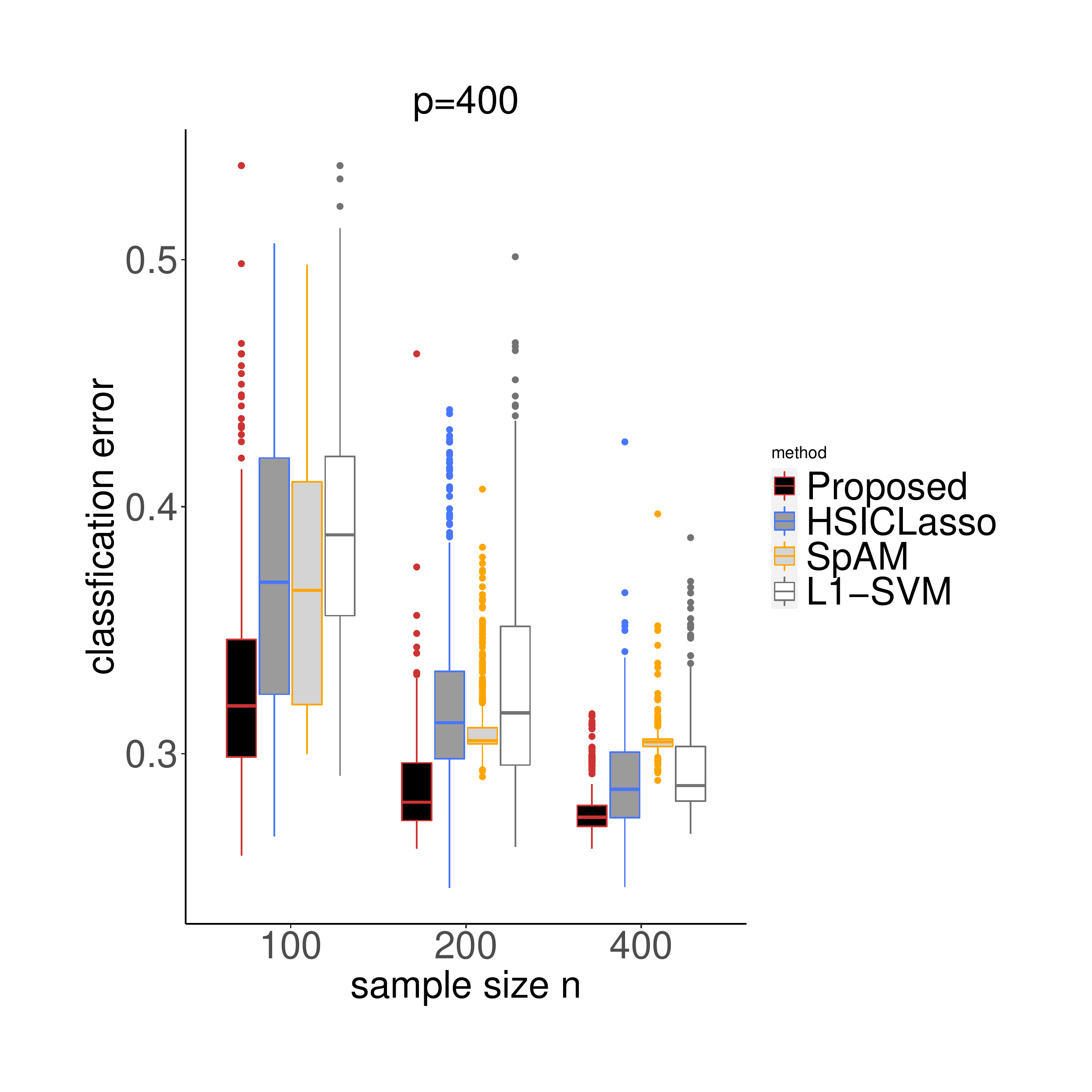}
\includegraphics[width=3.3in, height=2.5in]{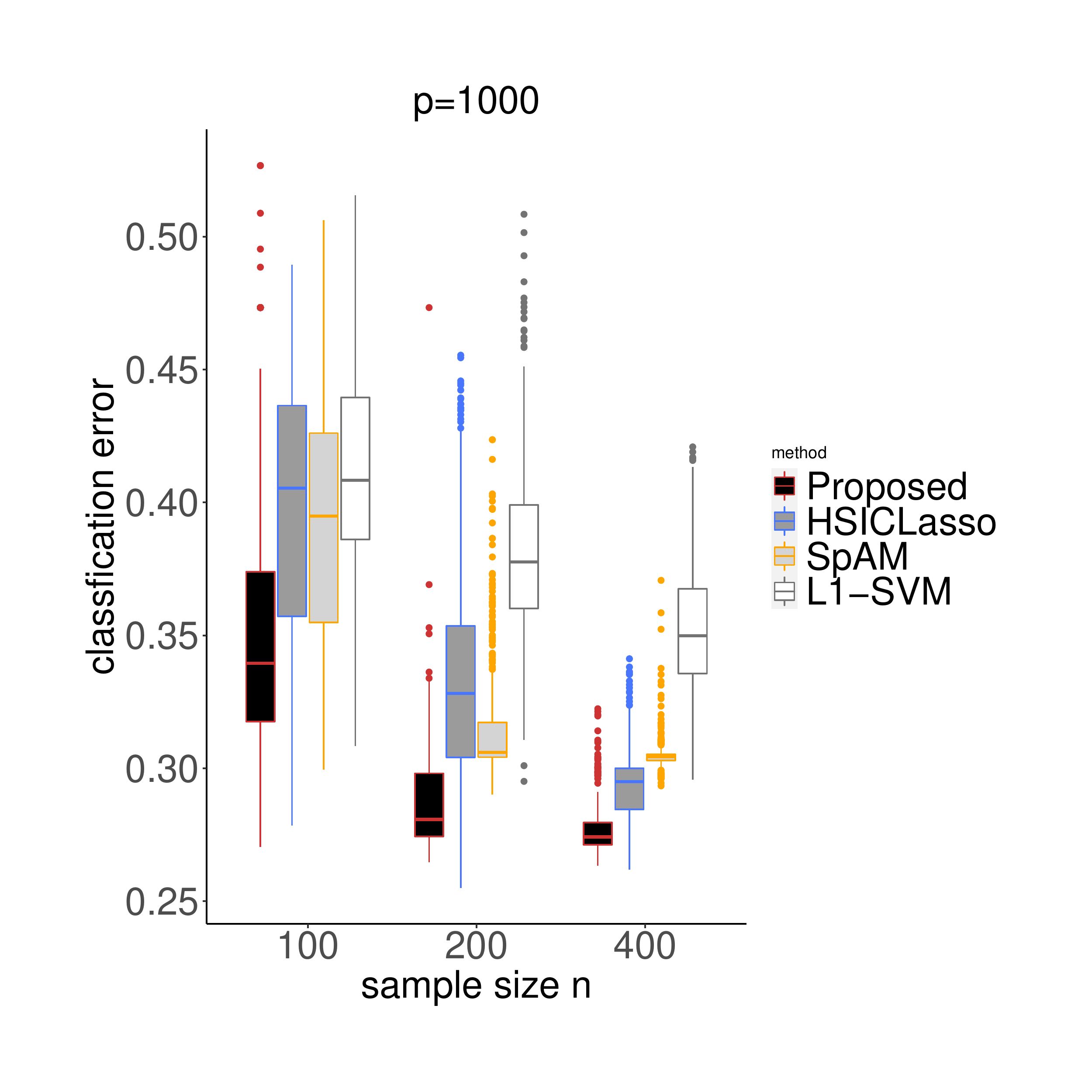}
\end{center}

{Note. The plots give the distribution of misclassification rates among four competing methods. The comparing methods from left to right in each plot are our proposed method, HSICLasso, SpAM and $l_1$-SVM.}
\end{figure}

\section{Application}
We applied our proposed method to analyze a  gene expression study in \citet{Scheetz2006}. This study analyzed microarrays RNAs  of eye disease from 120 male rats, containing the expression levels from about $31,000$ gene probes.  One interesting question was to determine which probes might be associated with the expression of gene TRIM32, which had been implicated in a number of diverse biological pathways and also known to be one of 14 genes linked to Bardet-Biedl syndrome \citep{TRIM32}.
For this purpose, we dichotomized TRIM32  based on whether it was over expressed as compared to a reference sample in the dataset. We further restricted our feature variables to the top 1000 probe sets that were most correlated with TRIM32. All feature variables were on a log-scale and standardized in the analysis.
To examine the performance of our method, we randomly divided the whole sample so that $70 \%$ were used for training and the rest were used for testing. This random splitting was then repeated 500 times to obtain reliable results.
For each training data, we used 3-fold cross validation to choose tuning parameters. We also applied HSICLasso, SpAM and $l_1$-SVM for comparison.

The analysis results are shown in Table 3. We notice that our method gives almost the same classification error as $l_1$-SVM, which is the smallest on average. However,  our method selects a much smaller set of feature variables with an average of 5 variables. SpAM selects 13 variables on average but its classification error is higher.
In Table 4,  we report the top 10 most-frequent selected features among all 500 replications for each method. We notice that some features such as Fbxo7 and LOC102555217 were selected by at least three methods. In addition, 
Gene Sirt 3 was identified by all three nonlinear feature selection methods, but not  $l_1$-SVM, indicating some possible nonlinear relationship between Sirt 3 and TRIM32. 
In fact, Figure 3  reveals some nonlinear relationship between Sirt 3 and Fbxo7 using 5-Nearest-Neighbors model. Our method also selected some genes that were not identified by any other method.   We  applied our method to analyze the whole sample and obtained a training error of $21.9 \%$  along five 5 genes identified (Fbxo7, Plekha6, Nfatc4, 1375872 and 1388656), which were all selected as the top 10 genes in the previous random splitting experiment. 
\begin{table}
\caption{Summary of Feature Selection Results in The Real Data Application}
\begin{center} 
\begin{tabular}{l*{5}{c}}
\hline
 & min \# & max \# & avg \# & classification error \\ \hline
Proposed  Method   & 2      & 13      & 5.1     & 0.286 (0.057)                 \\ 
HSICLasso      & 1      & 1000      & 250.3     & 0.293 (0.046)              \\ 
SpAM      & 1      & 26     & 12.3      & 0.316 (0.057)                    \\ 
$l_1$-SVM     & 7      & 990      & 448.7      & 0.283 (0.058)                   \\ \hline
\end{tabular}
\end{center} 
{Note. The numbers are the mean  of misclassification rates from 500 replicates. The numbers within parentheses are the median absolute deviations from 500 replicates. ``min$\#$" is the minimum number of the selected features, ``max$\#$" is the max number of the selected features, and ``avg.$\#$" is the average number of  the selected features.}
\end{table}

\begin{table}
\caption{Top 10 Most Selected Genes for Each Method Based on 500 Random Splittings }
\begin{center}
\begin{tabular}{cccc}
\hline
               Proposed Method & HSICLasso & SpAM &  $l_1$-SVM \\
\hline
\multirow{10}{*}{}
                                             \makecell{{{\color{red}{\bf Fbxo7}}} \\ (67.5$\%$) }  &  \makecell{ Ska1 \\ (76.6$\%$)}  &   \makecell{1388491 \\ (46.2$\%$) } &   \makecell{ 1376747\\ (99.1$\%$) } \\
                  
                                              \makecell{{{\color{blue}{\bf Plekha6}}} \\(47.3$\%$) } &   \makecell{{{\color{green}{\bf Sirt3}}} \\ (76.2$\%$) }  & \makecell{ {{\color{red}{\bf Fbxo7}}}  \\ (37.8$\%$) }   &  \makecell{{{\color{cyan}{\bf 1390538}}}\\ (98.9$\%$) } \\
                                  
                                             \makecell{{{\color{brown}{\bf LOC102555217}}} \\ (24.5$\%$)}    &  \makecell{ Ddx58 \\ (76.2$\%$) }  &  \makecell{Slco1c1 \\ (36.6$\%$) }  &   \makecell{ RragB  \\ (98.6$\%$) } \\
                                 
                                               \makecell{Nfatc4 \\ (22.7$\%$) } &  \makecell{1371610 \\ (76.0$\%$)}   & \makecell{ Stmn1 \\ (35.4$\%$)  } &   \makecell{ Atl1 \\ (97.9$\%$) }  \\
                                          
                                              \makecell{{{\color{cyan}{\bf 1390538}}} \\ (20$\%$) }  &  \makecell{ LOC100912578 \\ (73.2$\%$)}   & \makecell{1373944 \\ (32.4$\%$) }      &   \makecell{ {{\color{red}{\bf Fbxo7}}} \\ (97.3$\%$) }     \\
                                  
                                           \makecell{{{\color{magenta}{\bf 1375872}}}\\ (20$\%$) }  &   \makecell{Ttll7\\  (70.4$\%$) } & \makecell{  Ufl1 \\ (32.2$\%$) }
      &   \makecell{{{\color{blue}{\bf Plekha6}}}  \\(95.1$\%$)  }  \\
                             
                                              \makecell{{{\color{ teal}{\bf RGD1306148}}} \\(13.4$\%$) } &  \makecell{ Decr1 \\ (70.4$\%$) } & \makecell{ LOC100912578 \\ (31.0$\%$) }  &    \makecell{{{\color{magenta}{\bf 1375872}}}\\ (94.8$\%$)} \\
                                           
                                             \makecell{{{\color{green}{\bf Sirt3}}} \\(11.6$\%$) }  &   \makecell{Mff \\ (68.0$\%$) } & \makecell{ LOC100911357 \\ (28.6$\%$) } &  \makecell{{{\color{ teal}{\bf RGD1306148}}}\\ (94.1$\%$) }  \\
                                           
                                                \makecell{Prpsap2 \\(11.4$\%$) } &  \makecell{Pkn2 \\ (67.0$\%$) }   &\makecell{{{\color{brown}{\bf LOC102555217}}} \\ (26.8$\%$) }  &  \makecell{ Ska1 \\ (93.6$\%$) }  \\
                                           
                                                \makecell{1388656 \\(10.2$\%$) }&  \makecell{Taf11 \\(65.0$\%$) }  & \makecell{{{\color{green}{\bf Sirt3}}} \\ (22.4$\%$) }   &  \makecell{{{\color{brown}{\bf LOC102555217}}} \\ (93.2$\%$) } \\
                                          
                                           \hline
\end{tabular}
\end{center}
{Note. The numbers within parentheses are the frequencies to be selected in 500 random splittings. The genes also selected by the proposed method are highlighted in boldface.}
\end{table}


\begin{figure}
\caption{5-Nearest-Neighbor Plot of Sirt3 versus Fbxo7 in Real Data Study}
\centering
\includegraphics[scale = 0.5]{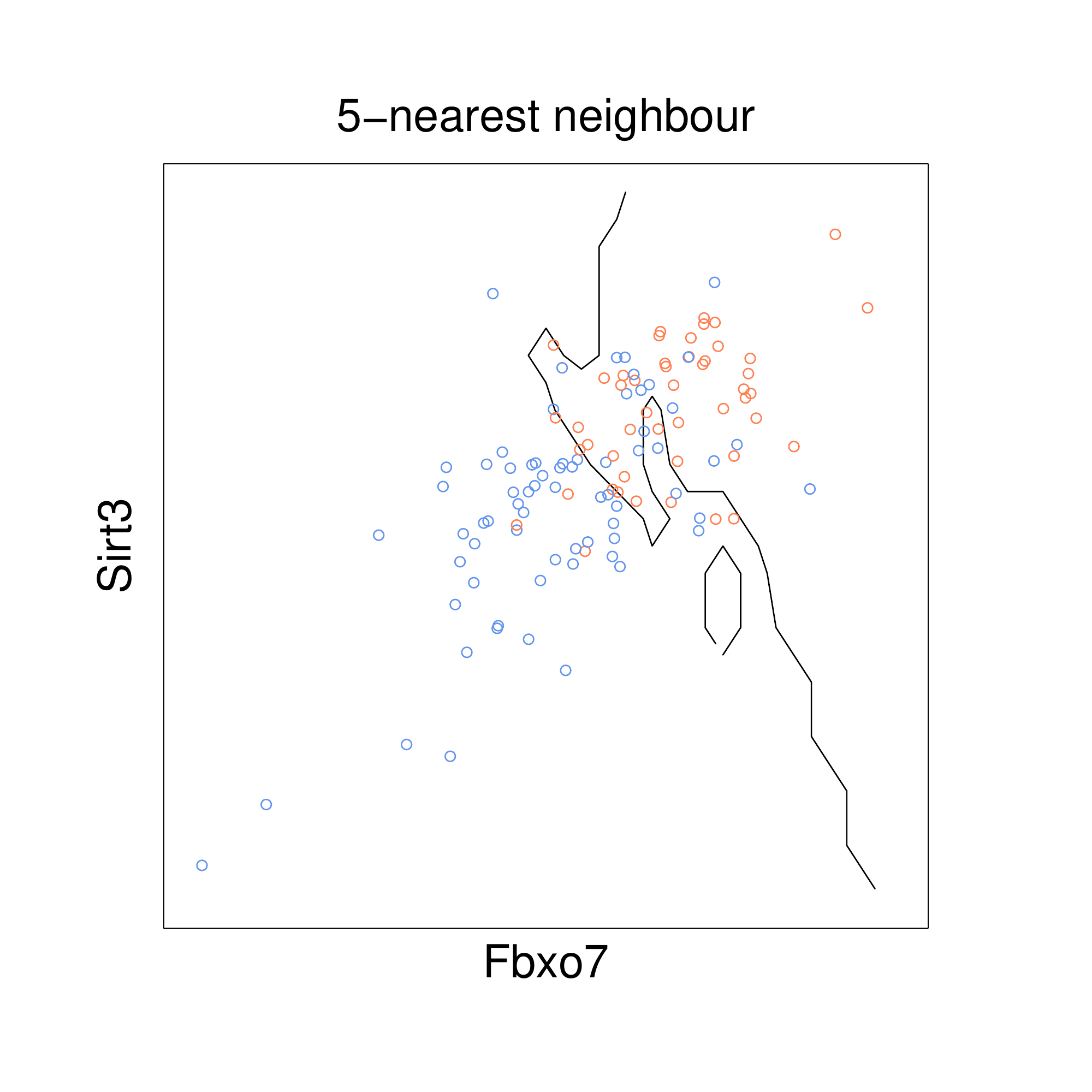}
\end{figure}

\section{Discussion}
In this work, we have proposed a general framework for nonparametric feature selection for both regression and classification in high dimensional settings. We introduced  a novel tensor product kernel for empirical risk minimization. This kernel led to fully nonparametric estimation for the prediction function but  allowed the importance of each feature to be captured by a non-negative parameter in the kernel function. Our approach is computationally efficient because it iteratively solves a convex optimization problem in a coordinate descent manner. We have shown that the proposed method has theoretical oracle property for variable selection. The superior performance of the proposed method was demonstrated via simulation studies and a real data application with a large number of feature variables.

We considered $l_2$ loss function for regression and exponential loss function for classification as examples. Clearly, the proposed framework applies to feature selection under many different loss functions in machine learning field. Another extension is to incorporate structures of feature variables in constructing the kernel function. For example, in integrative data analysis, feature variables arise from many different domains such as clinical domain, DNA, RNA, imaging and nutrition. It will be interesting to construct a hieachical kernel function which can not only identify feature variables within each domain but also identify important domains at the same time.

Our framework of nonparametric feature selection can be generalized to precision medicine where one of the main goals is to identify predictive biomarkers for treatment response. We can adopt loss functions used for precision medicine in our proposed method to simultaneously accomplish variable selection and discovering optimal individual treatment rules. Extensions to categorical outcomes and multi-stage treatment rule estimation are also possible under our general framework, which can be pursued in future work.

\bibliographystyle{biom}
\bibliography{ref}

\end{document}